\documentclass[11pt]{article}
\usepackage{setspace}
\usepackage{fullpage}

\usepackage{hyperref}
\usepackage{amsmath, amssymb, amsthm, amsfonts, graphicx, color, subcaption, enumerate, bm, cleveref, enumitem, array, mathtools, soul,mdframed, lineno}
\usepackage[ruled,linesnumbered]{algorithm2e}

\SetKwIF{OnEvent}{OnElseIf}{OnElse}{on event}{do}{on event else if}{else}{end}
\SetKwIF{WaitUntil}{WaitElseIf}{WaitElse}{wait until}{then}{else if}{else}{end}
\SetKwInOut{Requires}{Requires}
\crefname{algocf}{algorithm}{algorithms}
\Crefname{algocf}{Algorithm}{Algorithms}

\usepackage[style=trad-alpha,natbib=true,maxcitenames=4]{biblatex}
\usepackage{thm-restate}

\usepackage[colorinlistoftodos,prependcaption,textsize=tiny]{todonotes}

\newcommand{\leader}{\mathsf{Leader}}
\newcommand{\nonleader}{\mathsf{Non\text{-}Leader}}

\newcommand{\id}{\mathsf{ID}}

\newcommand{\lnode}{\mathsf{L}}
\newcommand{\mnode}{\mathsf{M}}
\newcommand{\xnode}{\mathsf{X}}

\newcommand{\ruleu}{$\mathbf{Upstream}$}
\newcommand{\rulel}{$\mathbf{Leader}$}
\newcommand{\ruled}{$\mathbf{Downstream}$}

\newcommand{\ps}{$\mathsf{Port*}$}

\newtheorem{theorem}{Theorem}[section]

\newtheorem{lemma}[theorem]{Lemma}

\newtheorem{observation}[theorem]{Observation}

\newtheorem{definition}[theorem]{Definition}
\title{Beyond 2-Edge-Connectivity: Algorithms and Impossibility for Content-Oblivious Leader Election}
\author{
  Yi-Jun Chang\footnote{National University of Singapore. ORCID: 0000-0002-0109-2432. Email: cyijun@nus.edu.sg} 
  \and Lyuting Chen\footnote{National University of Singapore. ORCID: 0009-0002-8836-6607. Email: e0726582@u.nus.edu}
  \and Haoran Zhou\footnote{National University of Singapore. ORCID: 0009-0001-2458-5344. Email: haoranz@u.nus.edu}
}

\date{}

\begin{document}

\maketitle

\begin{abstract}
The \emph{content-oblivious model}, introduced by \citeauthor{fully-defective-22} (PODC 2022; Distributed Computing 2023), captures an extremely weak form of communication where nodes can only send asynchronous, content-less pulses. They showed that in 2-edge-connected networks, any distributed algorithm can be simulated in the content-oblivious model, provided that a unique leader is designated \emph{a priori}. Subsequent works of \citeauthor{content-oblivious-leader-election-24} (DISC 2024) and \citeauthor{2-connected-leader-election-25} (DISC 2025) developed content-oblivious leader election algorithms, first for unoriented rings and then for general 2-edge-connected graphs. These results establish that all graph problems are solvable in content-oblivious, 2-edge-connected networks.

Much less is known about networks that are \emph{not} 2-edge-connected. \citeauthor{fully-defective-22} showed that no non-constant function \(f(x,y)\) can be computed correctly by two parties using content-oblivious communication over a single edge, where one party holds \(x\) and the other holds \(y\).
This seemingly ruled out many natural graph problems on non-2-edge-connected graphs.


In this work, we show that, with the knowledge of network topology $G$, leader election is possible in a wide range of graphs. Our main contributions are as follows:
\begin{description}
    \item[Impossibility:] \emph{Graphs} symmetric about an edge admit no randomized terminating leader election algorithm, even when nodes have unique identifiers and full knowledge of $G$.
    \item[Leader election algorithms:] \emph{Trees} that are not symmetric about any edge admit a quiescently terminating leader election algorithm with topology knowledge, even in anonymous networks, using $O(n^2)$ messages, where $n$ is the number of nodes. Moreover, even-diameter trees admit a terminating leader election given only the knowledge of the network diameter $D = 2r$, with message complexity $O(nr)$.
    \item[Necessity of topology knowledge:] In the family of graphs $\mathcal{G} = \{P_3, P_5\}$, both the 3-path $P_3$ and the 5-path $P_5$ admit a quiescently terminating leader election if nodes know the topology exactly. However, if nodes only know that the underlying topology belongs to $\mathcal{G}$, then terminating leader election is impossible.
\end{description}
\end{abstract}

\thispagestyle{empty}
\newpage
\tableofcontents
\thispagestyle{empty}
\newpage
\pagenumbering{arabic}
\section{Introduction}
\label{sect:intro}

Distributed systems often operate under unreliable communication, where messages may be corrupted. In many networks, silicon-based or natural, the inherent presence of \emph{alteration noise} means any \emph{fault-tolerant} algorithm must take received messages ``with a grain of salt'', instead of relying on their content verbatim, whose integrity could have been compromised in transmission. 

\paragraph{Content-oblivious model} To design a \emph{fault-tolerant} algorithm that operates correctly even when facing the most extreme form of alteration noise, the work of \citet{fully-defective-22} adopted the idea of using \emph{the sheer existence} of messages to encode information. To demonstrate the full power of their results, they first proposed the extremely weak model of \emph{fully-defective} networks, where message content is subject to arbitrary, unbounded alteration; therefore, such content-less messages are aliased \emph{pulses}, and any correct algorithm must completely ignore message content, hence \emph{content-oblivious}. Also, the network is \emph{asynchronous}, meaning that nodes lack access to a global clock; each pulse delivery takes an unpredictable time, thus forbidding the use of the presence/absence of pulses to encode information.

\paragraph{Universal solvability from 2-edge-connectivity} Naturally, one might conjecture that any non-trivial computational task is impossible when all messages are fully corrupted and arbitrarily delayed. However,~\citet{fully-defective-22} showed otherwise: Assuming the presence of a \emph{pre-elected} leader, any noiseless algorithm can be simulated in the content-oblivious setting, provided the network is \emph{2-edge-connected}. The key idea is that the leader communicates with a recipient using two disjoint paths—one to transmit the message in unary, and the other to signal the end of the transmission. This ensures that the recipient never waits indefinitely for the message to conclude, allowing the computation to make progress. Next, the work of \citet{content-oblivious-leader-election-24} showed how to elect a leader from scratch in an oriented ring as the simplest form of 2-edge-connected networks, while \cite{2-connected-leader-election-25} removed the orientation requirement. Moreover, the latter \cite{2-connected-leader-election-25} showed that even for general 2-edge-connected networks, with an upper bound on the network size $N$ known, non-uniform leader election exists. Combined with the simulation result of \citet{fully-defective-22}, their findings imply the solvability of all graph problems in a content-oblivious, 2-edge-connected network.

\paragraph{Beyond 2-edge-connectivity}
On graphs that are not 2-edge-connected, \citet{fully-defective-22} showed the first negative result in the \emph{vanilla} content-oblivious model. They showed that no non-constant function \(f(x,y)\) can be computed correctly and deterministically by two parties using content-oblivious communication over a single edge, where one party holds \(x\) and the other holds \(y\). This two-party impossibility immediately yields a general impossibility result for any network containing a bridge $e$, since we may let the two parties simulate the two connected components of $G-e$.

While this observation suggests that many natural graph problems---including leader election---might be impossible to solve on non-2-edge-connected graphs, it is not immediately clear what \emph{specific} impossibility results can be derived directly from the two-party impossibility. In this work, we take a different perspective and investigate which tasks \emph{can} still be performed on non-2-edge-connected graphs. Indeed, it may well be that many non-trivial tasks remain feasible in ways that are fully consistent with the two-party impossibility result.

\subsection{Contribution and Roadmap}
In this work, we show a \emph{complete characterization} of content-oblivious leader election on trees with topology knowledge: Terminating leader election is possible \emph{if and only if} the tree is not symmetric about an edge.

\begin{definition}[Edge symmetry]
    A graph $G$ is symmetric about an edge $e = \{u,v\}$ if $G-e$ has exactly two connected components, $H$ and $K$, containing $u$ and $v$ respectively, and there is an isomorphism $f:V(H) \rightarrow V(K)$ with $f(u) = v$. 
\end{definition}

If a graph $G$ is symmetric about an edge $e$, then $e$ is necessarily a bridge, so $G$ is non-2-edge-connected.
On the feasibility side, in \Cref{sect:algorithm}, we first design a content-oblivious leader election algorithm for even-diameter trees only (\Cref{thm:even_diameter_leader_election}); following that, we extend the result to all trees that are not symmetric about any edge (\Cref{thm:general_tree_leader_election}).

\begin{restatable}[Leader election on even-diameter trees]{theorem}{mainthmeven}
\label{thm:even_diameter_leader_election} 
There exists a terminating, anonymous content-oblivious leader election algorithm with message complexity $O(nr)$ for a tree with even diameter $D = 2r$ and $n$ nodes, provided that each node knows the diameter $D$ of the tree \emph{a priori}.
\end{restatable}

\begin{restatable}[Leader election on general asymmetric trees]{theorem}{mainthm}
\label{thm:general_tree_leader_election}
There exists a quiescently terminating, anonymous content-oblivious leader election algorithm with message complexity $O(n^2)$ for a tree with $n$ nodes, provided that each node knows tree topology $G$ \emph{a priori}, and the tree is not symmetric about any edge.
\end{restatable}

Turning to the impossibility side, in \Cref{sect:impossibility}, we show that for \emph{graphs} that are symmetric about an edge, topology knowledge does not result in any randomized algorithm that significantly outperforms trivial guesswork. 

\begin{restatable}[Edge symmetry implies impossibility, randomized]{theorem}{impossibility}
\label{thm: symmetric_about_edge}
Let $G$ be a network that is symmetric about an edge, where nodes are anonymous and \ul{know the network topology $G$}. Then there exists no terminating, randomized leader-election algorithm which succeeds with probability greater than $\frac{1}{2}$.
\end{restatable}

As a corollary, we obtain the analogous impossibility result in the deterministic setting where nodes have unique $\id$s. The proof of the corollary is also in \Cref{sect:impossibility}.

\begin{restatable}[Edge symmetry implies impossibility, deterministic]{corollary}{impossibilitydet}
\label{thm: symmetric_about_edge_det}
Let $G$ be a network that is symmetric about an edge, and nodes are equipped with unique $\id$s and \ul{know the network topology $G$}. Then there exists no terminating, deterministic content-oblivious leader election algorithm.
\end{restatable}

In \Cref{sect:necessity}, we show that the \emph{exact} knowledge of the network topology is indeed a necessary condition for terminating leader election on asymmetric trees.



\begin{restatable}[Necessity of topology knowledge]{theorem}{necessity}
There exists no deterministic, terminating content-oblivious leader election algorithm if the underlying topology is drawn from the family $\mathcal{G} = \{P_3, P_5\}$, and $\mathcal{G}$ is known to the nodes \emph{a priori}, even if the nodes have unique $\id$s.
\end{restatable}

Finally, in \Cref{sec:Stabilizing}, we present a simple \emph{stabilizing} leader election algorithm that works for all trees, showing that the \emph{termination} requirement is indeed \emph{necessary} for the impossibility results. 

\begin{restatable}[Stabilizing leader election]{theorem}{stabilizing}
\label{thm:stabilizing}
There exists a stabilizing content-oblivious leader election algorithm of message complexity $n+2\cdot\mathsf{ID_{max}}-1$ on a tree with $n$ nodes with unique $\id$s.
\end{restatable}

The algorithm of \Cref{thm:stabilizing} is of independent interest. Although the synchronized counting phase in the algorithm from the prior work~\cite{2-connected-leader-election-25} already yields a stabilizing leader election algorithm, it is \emph{non-uniform}, as it requires a known upper bound $N$ on the number of nodes~$n$. In contrast, the algorithm in \Cref{thm:stabilizing} is \emph{uniform}: it does not rely on any prior knowledge of the network or the identifier space. Moreover, for many parameter regimes, our message-complexity bound of $n + 2 \cdot \mathsf{ID_{max}}-1$ in \Cref{thm:stabilizing} compares favorably with the $O(N m \mathsf{ID_{min}})$ bound from the prior work~\cite{2-connected-leader-election-25}, where $m$ is the number of edges.

\paragraph{Technical overview} 
We briefly outline the main techniques underlying our leader election algorithms. Let $D$ denote the diameter of the tree. In an even-diameter tree, if we repeatedly remove the leaves of the remaining tree, then after $D/2$ steps, only a single node remains. Our algorithm for even-diameter trees (\Cref{thm:even_diameter_leader_election}) can be viewed as an implementation of this \emph{leaf-peeling} process in the content-oblivious model, where the main challenge is handling asynchrony using only pulses. The key idea is to use the number of pulses sent to encode the height of a node.

For odd-diameter trees, the above leaf-peeling process leaves exactly two nodes. To elect a leader between them, we exploit the assumption that the tree is asymmetric. Our algorithm for general asymmetric trees (\Cref{thm:general_tree_leader_election}) uses the number of pulses to encode the topology of a subtree, thereby breaking symmetry between the two remaining nodes.

Finally, for our stabilizing algorithm for all trees (\Cref{thm:stabilizing}), we simply trim leaves in an arbitrary order until only two nodes remain, and then break symmetry between them by comparing their identifiers.

\subsection{Additional Related Works}
Leader election is one of the most fundamental problems in distributed systems, and has therefore enjoyed a long history of study (e.g., \cite{10.1145/359104.359108,10.1145/359024.359029,10.1145/69622.357194,DBLP:journals/jal/DolevKR82,Itai90negative}). The works of \cite{boldi1996symmetry,yamashita1996computing} study the conditions and impossibilities of breaking symmetry in anonymous networks, while \cite{yamashita1989electing,dobrev2004leader} examine symmetry breaking when nodes are labeled but $\id$ collisions occur. A more recent result closely related to our work is that of \citet{time-vs-info-16}, who studied the role of uniform advice in leader election. However, since their setting was the $\mathsf{LOCAL}$ \emph{message-passing} model, advice served only to accelerate leader election rather than enable it, leading to a series of trade-offs between locality and advice size. By contrast, if advice is non-uniform, leader election becomes trivial when only a simple ``yes/no'' answer is required. For stronger definitions of leader election~\cite{four-shades} and for a broader range of information dissemination problems, non-uniform advice, also known as informative labeling schemes \cite{informative-labeling}, has been studied.

Similar to the content-oblivious setting, various other weak-communication models have been studied to capture systems with constrained communication capabilities. A notable example is the synchronous \emph{beeping model} introduced by~\citet{beeping-model-10}, where during each round, nodes can broadcast a beep to their neighbors and can distinguish between silence and the presence of at least one beep from their neighbors.

In the beeping model, the first positive result for leader election was shown by~\citet{ghaffari2013near}, who also established an $\Omega(D+\log n)$-round lower bound. 
Subsequent works~\cite{forster2014deterministic,beeping-time-optimal-18,emek2021thin} gradually approached this lower bound. In parallel, the works of~\citet{10.1007/978-3-662-48653-5_3} and~\citet{beeping-minimal-weak-communication-25} focused on a different direction: Designing algorithms that operate on a constant number of states, at a cost of a slightly higher round complexity.

The \emph{population protocol} introduced by~\citet{angluin2006computation} studies massive systems of state machines under random pairwise interactions. Assuming each pair of agents has equal probability of interaction, a series of results~\cite{doty2018stable,10.5555/3039686.3039855,sudo2020leader} explores the trade-off between the number of states and the number of expected interactions for electing a leader. The works of~\citet{berenbrink2020optimal} and~\citet{9086071} proposed optimal leader elections under different state-size constraints. A recent work of~\citet{alistarh2022near} extended the study to cover the setting where only a subset of agent pairs may interact.

The \emph{leaf-peeling} technique is well known and has been used in the design of various tree algorithms. For example, \citet{peeling99} employed this approach to develop distributed algorithms for finding the centers and medians of a tree network.

\section{Preliminaries} 
\paragraph{Content-oblivious model} Let $G = (V, E)$ be a communication network, where each node $v\in V$ is a computing device of unbounded computation power, and each edge $e\in E$ is a communication link. Messages exchanged by nodes are subject to unbounded alteration noise, therefore, effectively contentless. Equivalently, we say computing devices only exchange \emph{pulses}. All pulses are subject to \emph{unpredictable, unbounded delay}, despite that an eventual delivery is always guaranteed, and nodes are aware of the incoming port of a pulse received. Therefore, a notion of shared time \emph{does not} provide additional power to the network: the content-oblivious network is \emph{asynchronous}, meaning a computing device may only perform actions upon initiation of any algorithm, or upon receiving a pulse. Without loss of generality, we assume all local computations to be instantaneous. When measuring the message complexity of an algorithm, we consider the total number of pulses sent.

\paragraph{Model variants}
Nodes in a content-oblivious network may either have unique positive integer $\id$s or be anonymous. In the latter case, the behavior of a node is completely determined by the pulses it has received so far. The algorithm itself may be deterministic or randomized; in the randomized setting, each node is given an \emph{infinite} and \emph{independent} stream of random bits.

\paragraph{Topology knowledge}
We say that nodes are provided with topology knowledge of $G$ if every node receives the same encoding of the topology of $G$ \emph{without} any $\id$ information, so that no node can \emph{immediately infer its position} in $G$ from identifiers alone. However, a node may still deduce partial information about its position from its degree. For example, if $G$ contains exactly two nodes of degree~3, then those two nodes can infer from the topology description that they must be among these two positions. When $G$ is a tree, the underlying topology can be represented succinctly in $O(n)$ bits (e.g., using a balanced-parentheses encoding), where $n$ is the number of nodes.



\paragraph{Leader election}
We study the \emph{leader election} problem, in which each node must decide between the outputs $\leader$ and $\nonleader$, and an execution of a leader election algorithm is successful if eventually exactly one node outputs $\leader$ while all others output $\nonleader$. We say that an algorithm \emph{stabilizes} if, after some finite time from the initiation of the algorithm, the outputs of all nodes remain fixed. This notion of stabilization is purely analytic and global: nodes are not required to detect that they have stabilized.

A stronger guarantee is \emph{termination}, which means that, after some finite time, each node \emph{commits} to an output and halts. Termination is said to be \emph{quiescent} if any node that has halted no longer receives pulses from its neighbors.

While stabilization ensures eventual global correctness, it requires nodes to continue executing and listening for pulses indefinitely. This behavior is undesirable in resource-constrained environments, such as those with limited energy or bandwidth, which often coincide with high-noise settings. In contrast, termination is locally detectable, strengthens correctness guarantees, and avoids the overhead associated with perpetual activity.

\section{Terminating Anonymous Leader Election on Trees}\label{sect:algorithm}

In this section, we first prove \Cref{thm:even_diameter_leader_election} by providing a terminating algorithm for even-diameter trees, where network diameter $D$ \textit{is provided to nodes a priori}. Following that, we generalize the above result to prove \Cref{thm:general_tree_leader_election} by providing a quiescently terminating algorithm where the underlying graph topology $G$ \textit{is provided to nodes a priori}.


\subsection{Terminating Anonymous Leader Election on Even-Diameter  Trees}

In this section, we show a terminating leader election in anonymous, even-diameter trees. This algorithm requires each node to know the diameter $D=2r$.

\mainthmeven*


We first provide a high-level overview of the algorithm. Each node performs a deterministic preprocessing phase locally using the information of $D$ to derive a set of rules that the node uses for the upcoming election phase. Nodes hence monitor the triggering conditions of all rules, and perform (sending pulses and outputting leader election results) accordingly as the rules dictate. The pulse movement follows a bottom-top-bottom pattern. Leaf nodes initiate the election by sending pulses to their neighbors, and each node monitors the number of pulses it receives from each edge, constantly checking them against the set of rules derived. The center, which is the node with the minimum eccentricity, is unique for even-diameter trees. It triggers a special \rulel~rule and terminates as leader. It then broadcasts a pulse so that every other node terminates as non-leader.

In the above high-level description, we mentioned nothing about what a rule looks like or how to interpret a rule. To understand that, let us first provide an intuition of what pulses encode in our leader election.

A certain number of pulses sent from node $u$ to node $v$ along the edge $e=\{u,v\}$ encodes the \emph{height of the subtree} rooted at $u$, that is, the connected component of $u$ in $G - e$. A subtree with height $i$, where $0\leq i \leq r-1$, is encoded by $r-i$ pulses. To be elected as leader, a node must receive at least one pulse from each neighbor. This seemingly upside-down encoding is to ensure two important properties: 
\begin{itemize}
    \item \textbf{The center is the only node that can become a leader.}
If we root the tree at any vertex \(v \neq l\), then there exists a root-to-leaf path whose length exceeds \(r\). Consequently, there is a subtree rooted at a child $u$ of $v$ whose height exceeds $r-1$. In this case, \(v\) can never receive a valid pulse from \(u\) encoding that subtree, and thus can never claim leadership. This ensures the \emph{correctness} of the leader election. 
    \item \textbf{Before a leader is elected, all pulses travel in the direction from leaves to the center.} This property ensures that after a leader is elected, it can broadcast a pulse that unambiguously signals that a leader has been elected, hence making all other pulses terminate as non-leaders, ensuring the \emph{termination} of all nodes.
\end{itemize}




We are now ready to formally describe the algorithm. 
Each node derives three categories of rules from the knowledge of $D=2r$: \ruleu, \rulel, and \ruled. A node keeps track of the number of received pulses on each port and compares against the rules upon each update. Initially, only \ruleu~and \rulel~rules are active. After a node triggers at least one of \ruleu~and \rulel~rules, \ruled~rules become active for that node.

\begin{mdframed}
\begin{itemize}
    \item \ruleu: There are $r$ copies of \ruleu, one for each $i\in[1,r]$. The $i$th rule is defined as follows: 

    Whenever a node with $d$ ports sees $d-1$ of its ports (that is, all except one) each receive at least $i+1$ pulses while the remaining port receives \emph{none}, mark the remaining port as \ps. Send pulses along \ps, until the \emph{total number of pulses} sent through \ps, since the start of the algorithm, is exactly $i$. Add \ruled~to the list of active rules, and remove the \rulel~rule.

    Note that if a node has one port only, it automatically triggers the \ruleu~rule constructed for all $i \in[1,r]$ upon initiation, resulting in sending $r$ pulses along its only port, which is marked \ps.
    
    \item \rulel: Whenever a node with $d$ ports sees each of its ports receive at least one pulse, send a pulse along all ports, declare $\leader$ and terminate.
    
    \item \ruled: When a pulse is received from $\mathsf{Port*}$, send a pulse along every port except for $\mathsf{Port*}$, declare $\nonleader$, and terminate.
\end{itemize}
\end{mdframed}

To facilitate the analysis, we consider the \emph{layer decomposition} resulting from the leaf-peeling process, defined as follows.

\begin{algorithm}
\DontPrintSemicolon
\caption{Layer decomposition} \label{alg: layer_decomp}
$i \gets 0$\;
$G_0 \gets G$\;
\While {$G_i$ is not empty}{
    $V_i \gets \text{set of leaf nodes of } G_i$\;
    $G_{i+1} \gets G_i \setminus V_i$\;
    $i \gets i+1$\;
}
\end{algorithm}

Since the radius of the network is \(r\), the last non-empty layer \(V_r\) contains exactly one node (the center of the tree), which we denote by \(l\) in the subsequent discussion.

\paragraph{Example} As a concrete example, consider a complete binary tree of radius $r$. The leader election proceeds as follows: First, by the \ruleu~rule, each node in $V_0$ (i.e., leaves) sends $r$ pulses to its parent, who is in $V_1$. Each node in $V_1$ therefore receives $r$ pulses from its neighbors in $V_0$, and by the \ruleu~rule, sends $r-1$ pulses to its parent, etc. Performing induction upwards, eventually the root receives at least one pulse from each of its neighbors, so it becomes the leader by triggering the \rulel~rule. After that, the \ruled~rule will be triggered top-down: the now elected leader sends a pulse to each of its children, and then each node that receives a pulse from its parent triggers the \ruled~rule, sending a pulse to each of its children, until the entire tree is reached.

\paragraph{Analysis} While the preceding example provides intuition for the correctness of the algorithm on all even-diameter trees, we now present a formal proof that the constructed rules indeed yield a terminating leader election procedure. 



\begin{observation}
    For each $v\in V_i$ with $0\leq i\leq r-1$, it has exactly one neighbor $u$ in $G$ that belongs to $\bigcup_{j > i} V_j$, while all other neighbors are in $\bigcup_{j < i} V_j$. For the unique node $l \in V_r$, it has all neighbors in $\bigcup_{j < r} V_j$.
\end{observation}

\begin{proof}
For any vertex $v \in V_i$, the vertex $v$ must be removed in the $i$th round of the layer decomposition. Thus, $v$ must be a leaf in $G_i$, which implies that all but possibly one of its neighbors in $G$ have already been removed prior to the $i$th round and therefore lie in $\bigcup_{j < i} V_j$.
    
If $0 \le i \le r-1$, then $v$ must have exactly one neighbor in $G_i$ that is not removed in the $i$th round. Otherwise, $G_{i+1}$ (and hence $V_{i+1}$) would be empty, implying that $V_i$ is the final non-empty layer, so $i = r$, contradicting $0 \le i \le r-1$. Therefore, $v$ has exactly one neighbor in $G$ that lies in $\bigcup_{j > i} V_j$. 
\end{proof}

\begin{definition} [Parent and children]
\label{def:parent and children}
    For each $v\in V_i$, define the (possible) neighbor of $v$ in $G$ belonging to $\bigcup_{j > i} V_j$ as the \ul{parent} of $v$. Define the neighbors of $v$ in $G$ belonging to $\bigcup_{j < i} V_j$ as the \ul{children} of $v$. 
\end{definition}

Clearly, the unique node $l \in V_r$ has no parent, and all nodes $v\in V_0$ have no children.

\begin{observation} \label{obs:one child in layer i-1}
    For each $v\in V_i$ with $1\leq i\leq r$, it has at least one child that belongs to $V_{i-1}$.
\end{observation}

\begin{proof}
    For the sake of contradiction, assume that all neighbors of $v$ in $\bigcup_{j < i} V_j$ also lie in $\bigcup_{j < i-1} V_j$. This would imply that $v$ is a leaf in $G_{i-1}$ and therefore should have been removed in round $i-1$, placing it in $V_{i-1}$. This contradicts the assumption that $v \in V_i$.
\end{proof}



\begin{lemma} \label{lem: l two children}
    Node $l$ has at least two children in $V_{r-1}$.
\end{lemma}

\begin{proof}
Assume otherwise that $l$ has only one child in $V_{r-1}$. Then both $V_{r-1}$ and $V_r$ contain exactly one node. Consequently, $G_{r-1}$ is a path of two nodes. This is impossible, since in such a case both nodes would be removed in round $r-1$, contradicting the existence of a node in $V_r$.
\end{proof}


\begin{restatable}{observation}{portconstant} \label{obs: port* constant}
    Once a node triggers an \textbf{Upstream} rule, it never changes $\mathsf{Port*}$.
\end{restatable}

\begin{proof}
    To assign a port (without loss of generality, port 0) as \ps, a node must trigger an \ruleu\ rule at a time when port 0 has not yet received any pulses. Hence, no other port can have been assigned \ps\ earlier, since doing so would require port 0 to have already received a nonzero number of pulses in order to trigger some \ruleu\ rule.  
\end{proof}

\begin{lemma} \label{lem: at most r-i pulses to parent}
    If a node $u$ is in $V_i$, $i \in [0,r-1]$, then $u$ sends at most $r-i$ pulses to its parent.
\end{lemma}

\begin{proof}    
 We proceed by induction over the layers from $V_0$ to $V_r$, i.e., in a bottom-up manner.

\paragraph{Base case}
Let $u$ be a leaf node in $V_0$. By the definition of \textbf{Upstream}, $u$ sends $r$ pulses to its parent immediately upon initialization and never sends any more thereafter. Thus, the lemma holds for $u$.

\paragraph{Induction step}
Assume the lemma holds for all nodes in layers up to $V_s$ for some $s \in [0, r-2]$. Let $u$ be a node in $V_{s+1}$, and let $v$ be its parent.

Suppose that $u$ sets $\mathsf{Port*}$ to point toward $v$ (otherwise, the lemma holds trivially for $u$). By \Cref{obs:one child in layer i-1}, $u$ has exactly one child $m$ in layer $V_s$. By the induction hypothesis, $m$ sends at most $r - s$ pulses to $u$. By the definition of the \ruleu\ rules, $u$ can therefore send at most $r - s - 1$ pulses through \ps\ to its parent $v$.
\end{proof}

\begin{lemma} \label{lem: even diameter leader cannot upstream}
    Node $l$ cannot trigger any \ruleu~rule.
\end{lemma}

\begin{proof}
    By \Cref{lem: l two children}, node $l$ has at least two children in $V_{r-1}$, each sending at most one pulse to $l$ by \Cref{lem: at most r-i pulses to parent}. By the definition of \ruleu, $l$ cannot trigger any \ruleu~rule.
\end{proof}

\begin{lemma} \label{lem: even diameter send to parent before receiving}
    If a node $u$ is in $V_i$ with $i \in [0,r-1]$, then $u$ does not receive any pulses from its parent $v$ before sending one to $v$.
\end{lemma}

\begin{proof}
    We proceed by induction over the layers from $V_{r-1}$ to $V_0$, i.e., in a top-down manner.

    \paragraph{Base case} Consider any $v \in V_{r-1}$, whose parent is $l$. Suppose, for the sake of contradiction, that $v$ receives a pulse from $l$ before sending any pulse to $l$. Such a pulse cannot result from a triggered \ruleu\ rule at $l$, by \Cref{lem: even diameter leader cannot upstream}. It also cannot arise from a triggered \rulel\ rule, since that rule requires $l$ to have received pulses from \emph{all} of its ports, yet $v$ has not sent any pulse to $l$. This shows that $v$ cannot receive a pulse from $l$ before sending one. 

    \paragraph{Induction step} Assume that the lemma holds for all nodes in $V_s$ and higher layers for some $s \in [1, r-1]$. Let $u$ be a node in $V_{s-1}$ and $v$ its parent. Suppose, for the sake of contradiction, that $u$ receives a pulse from $v$ before sending one to $v$. Similar to the base case, $v$ cannot trigger the \rulel\ rule. Thus, it remains to rule out the possibility that $v$ triggers an \ruleu\ rule and sets the port pointing to $u$ as \ps.

For $v$ to trigger an \ruleu\ rule, it must have received a nonzero number of pulses from its parent $w$. By the induction hypothesis, $v$ must have sent pulses to $w$ before receiving any, and such pulses can only result from an \ruleu\ rule with the port pointing to $w$ designated as \ps. This, however, contradicts \Cref{obs: port* constant}. Hence, the assumed scenario is impossible, and $u$ also satisfies the lemma.
\end{proof}

The following three lemmas follow directly from \Cref{lem: even diameter send to parent before receiving}. 

\begin{restatable}{lemma}{onlylmaybeleader} \label{lem: even diameter only l may be leader}
    The only node that can trigger the \rulel~rule is $l$.
\end{restatable}

\begin{proof}
    To trigger the \rulel~rule, a node $v$ must receive pulses from all neighbors before sending any. For any node $v$ with a parent, the above is impossible due to \Cref{lem: even diameter send to parent before receiving}, so the only node that can trigger the \rulel~rule is $l$, as it is the only node without a parent.
\end{proof}

\begin{restatable}{lemma}{portpointsparent} \label{lem: even diameter port* points to parent}
    For each node $v$ that has a parent, its \ps\ (if any) may only point to its parent.
\end{restatable}

\begin{proof}
Suppose \ps\ points to a child. For this to occur, $v$ must have triggered an \ruleu\ rule in which it received pulses from its parent. By \Cref{lem: even diameter send to parent before receiving}, $v$ must then have previously sent a pulse to its parent, and this pulse cannot be the result of the \rulel\ rule by \Cref{lem: even diameter only l may be leader}. Thus, $v$ must have triggered an \ruleu\ rule and sent pulses to its parent earlier, which contradicts \Cref{obs: port* constant}.  
\end{proof}

\begin{restatable}{lemma}{leaderbeforedownstream} \label{lem: even diameter leader before downstream}
    No \ruled~rule is triggered before $l$ triggers the \rulel~rule.
\end{restatable}

\begin{proof}
    We proceed by induction over the layers from $V_{r}$ to $V_0$, i.e., in a top-down manner.

    \paragraph{Base case} The unique node $l$ in $V_r$ never triggers any \ruleu\ rule (\Cref{lem: even diameter leader cannot upstream}), and therefore \ruled\ is never activated at $l$.


    \paragraph{Induction step} Assume that the lemma holds for all nodes in $V_s$ and higher layers for some $s \in [1, r]$. Let $u$ be a node in $V_{s-1}$ and let $v$ be its parent. For $u$ to trigger \ruled, it must receive a pulse from \ps, which, by \Cref{lem: even diameter port* points to parent}, means a pulse from its parent $v$. Such a pulse cannot result from $v$ triggering \ruleu\ or \rulel, by \Cref{lem: even diameter port* points to parent} and \Cref{lem: even diameter only l may be leader}, respectively. Hence, $v$ must have triggered \ruled, and by the induction hypothesis, this implies that \rulel\ has already been triggered by $l$.
\end{proof}

Combining the above three lemmas, we see that \emph{the direction of pulse movement is always from children to parents}, before $l$ triggers the \rulel~rule.

\begin{restatable}{lemma}{pulsedirection} \label{lem: even diameter pulse direction}
    The direction of pulse movement is always from children to parents, before $l$ triggers the \rulel~rule.
\end{restatable}

\begin{proof}
By \Cref{lem: even diameter only l may be leader}, no node other than $l$ ever triggers the \rulel\ rule. Moreover, by \Cref{lem: even diameter leader before downstream}, no node other than $l$ triggers the \ruled\ rule before $l$ triggers the  \rulel\ rule. Thus, the \ruleu\ rules are the only rules that nodes other than $l$ may trigger before $l$ triggers the \rulel\ rule. By \Cref{lem: even diameter port* points to parent}, any pulse generated by such a \ruleu\ rule always travels toward the parent of the sender.
\end{proof}

\begin{lemma} \label{lem: even diameter tree success}
    Node $l$ eventually terminates as $\leader$ by triggering the \rulel~rule, after every other node sets \ps~to its parent.
\end{lemma}

\begin{proof}



   We again apply induction over the layers from $V_0$ to $V_r$. Our goal is to show that every $v \in V_i$ sends \emph{at least} $r - i$ pulses to its parent. In particular, this implies that each child of $l$ sends at least one pulse to $l$, since all of them belong to $V_{r-1}$ or a lower layer. Consequently, $l$ will trigger the \rulel\ rule.

    

    \paragraph{Base case} It is immediate that $u \in V_0$ sends $r$ pulses to its parent, by the definition of \ruleu.

    \paragraph{Induction step} Assume the lemma holds for all nodes up to $V_s$ for some $s \in [0, r-1]$. Let $u$ be a node in $V_{s+1}$. By the induction hypothesis, every child of $u$, all of which lie in $V_s$ or a lower layer, sends at least $r - s$ pulses to $u$. By the definition of the \ruleu\ rules, $u$ will therefore eventually send $r - s - 1$ pulses to its parent. We emphasize that the correctness of this induction step relies on \Cref{lem: even diameter pulse direction}, which ensures that we need only consider pulses traveling from children to parents.
\end{proof}

The above proof, together with \Cref{lem: at most r-i pulses to parent},
implies that every \(v \in V_i\) sends \emph{exactly} \(r - i\) pulses to its parent. The
argument is divided into two lemmas rather than one because the induction in the proof of
\Cref{lem: even diameter tree success} relies crucially on
\Cref{lem: even diameter pulse direction}, which itself depends on
\Cref{lem: at most r-i pulses to parent}.

\begin{lemma} \label{lem: even diameter tree success nonleader}
    Each node other than $l$ eventually terminates as $\nonleader$.
\end{lemma}

\begin{proof}
By \Cref{lem: even diameter pulse direction}, every non-$l$ node sets \ps\ to point to its parent. Immediately after $l$ triggers \rulel\ (which occurs by \Cref{lem: even diameter tree success}), each child of $l$ receives a pulse through its \ps\ and therefore triggers \ruled. Any node that triggers \ruled\ sends a pulse along every port except \ps, so each of its children likewise receives a pulse through its own \ps\ and triggers \ruled. This recursive propagation ensures that every non-$l$ node eventually triggers \ruled\ and terminates as $\nonleader$.
\end{proof}

\begin{lemma} \label{lem: even diameter complexity}
    The total number of pulses sent by all nodes in this algorithm is $O(nr)$.
\end{lemma}

\begin{proof}
Each node other than $l$ sends at most $r$ pulses to its parent via the \ruleu\ rules (see \Cref{lem: even diameter pulse direction}), contributing at most $(n-1)r$ pulses in total. In addition, the nodes send exactly $n-1$ pulses due to the \rulel\ and \ruled\ rules, since exactly one such pulse is sent across each edge and a tree has $n-1$ edges (see the proof of \Cref{lem: even diameter tree success nonleader}). Therefore, the total number of pulses sent is at most $(n-1)r + (n-1) = O(nr)$.
\end{proof}

\mainthmeven*
\begin{proof}[Proof of \Cref{thm:even_diameter_leader_election}]
    The correctness and termination of the leader election follow from \Cref{lem: even diameter tree success,lem: even diameter tree success nonleader}, while the message complexity follows from \Cref{lem: even diameter complexity}.
\end{proof}

\paragraph{Failing the quiescence guarantee}
The algorithm above does not guarantee quiescent termination. Recall from the proofs of \Cref{lem: at most r-i pulses to parent,lem: even diameter tree success} that any node \(u \in V_{s+1}\), for \(s \in [0, r-1]\), will eventually send exactly \(r - s - 1\) pulses to its parent. For this to occur, each child of \(u\) must have sent \(u\) \(r - s\) pulses. However, some children of \(u\) may lie in layers strictly below \(V_s\), and thus will ultimately send \(u\) more than \(r - s\) pulses. Crucially, \(u\) cannot determine whether such additional pulses will arrive in the future. Consequently, some pulses from the children  of  \(u\) may reach \(u\) after  \(u\)  has already terminated, and therefore quiescence cannot be guaranteed.


\subsection{Quiescently Terminating Anonymous Leader Election on Asymmetric Trees}

In this section, we show a quiescently terminating leader election in anonymous trees, unless they are symmetric about an edge. This algorithm, however, requires each node to know the full topology of $G$. In the remainder of the section, we will make the distinction between nodes, which are real computing devices performing content-oblivious leader election, and vertices of the topology knowledge $G$.

In contrast with the algorithm for even-diameter trees, the main challenge here is to break symmetry within layer \(V_r\), which contains two vertices when the tree diameter is odd. To resolve this symmetry, we use the assumption that \(G\) is not symmetric about an edge, implying that the two nodes in \(V_r\) have distinct views of the lower-layer nodes. Rather than using the number of pulses to encode only the \emph{height} of a subtree, our algorithm uses the number of pulses to encode the full \emph{shape} of each subtree. A pleasant consequence of this more refined encoding is that the algorithm now achieves quiescent termination.

\mainthm*

The remainder of this section is organized as follows. In \Cref{sect:preprocessing}, we detail the layer decomposition, enumeration of distinct subtrees, and construction of rules. In \Cref{sect:correctnessproof}, we prove that the above-constructed rules indeed result in a quiescently terminating leader election.

\subsubsection{Algorithm Description} \label{sect:preprocessing}

Each node first identifies the center(s) of the advised topology \(G\). To do so, every node locally performs the layer decomposition (\Cref{alg: layer_decomp}) using the advice \(G\), thereby partitioning the vertex set \(V\) into disjoint layers \(V_0, V_1, \ldots, V_r\). We emphasize that, unlike in the even-diameter algorithm, where the layer decomposition is used only in the analysis, here \emph{each node actually performs the layer decomposition locally}. 

\begin{definition} [Subtrees]
    For each vertex $v \in V_i$, define $G_{\leq v}$ as the induced subgraph of $G$ over the vertex set $\{v\} \cup (\bigcup_{j < i} V_j)$. Let $T^v$ be the connected component of $G_{\leq v}$ that contains $v$. We call $T^v$ \ul{the subtree rooted at $v$}.
\end{definition}

We first enumerate all distinct subtrees in an ``increasing'' order, and denote them by 
\(T_1, T_2, \ldots\) in the order in which they are listed. We stress the distinction between
subscripts and superscripts: we use a superscript \(v\) to denote the subtree rooted at \(v\),
i.e., \(T^v\), and a subscript \(i\) to index the enumeration, i.e., \(T_i\).

Each (rooted) subtree isomorphic class is enumerated only once. Note that if \(v \in V_i\),
then \(T^v\) has height \(i\); hence, two subtrees can be isomorphic only if their roots lie 
in the same layer. Conceptually, to carry out the enumeration, it suffices to define a
comparator on subtrees. The rules of this comparator are specified as follows.

\paragraph{Subtree comparator} The subtree comparator takes two non-isomorphic subtrees \(T^v\) and \(T^u\) as input and
determines whether \(T^v \prec T^u\) or \(T^u \prec T^v\). The enumeration of subtrees is
required to be consistent with this comparator: if \(T^v \prec T^u\), then \(T^v\) must be
listed before \(T^u\). We will later show that the comparator is transitive, ensuring that a
valid enumeration indeed exists. For isomorphic subtrees, we write \(T^v = T^u\), and we use
\(T^v \preceq T^u\) to denote either \(T^v = T^u\) or \(T^v \prec T^u\).

Given inputs $T^v$ and $T^u$, the comparator follows the three rules below.

\begin{description}
    \item[Low layer first:] Assume $v$ and $u$ are not of the same layer. If $v$ is from a lower layer, then the comparator outputs $T^v \prec T^u$, and vice versa.
    \item[Fewer children first:] Assume $v$ and $u$ are of the same layer. Let $N_{T^v}(v)$ (resp.,~$N_{T^u}(u)$) be the set of neighbors of $v$ (resp.,~$u$) in $T^v$ (resp.,~$T^u$). If $|N_{T^v}(v)| < |N_{T^u}(u)|$, then the comparator outputs $T^v \prec T^u$, and vice versa.
    \item[Recursion rule:] Assume \(v\) and \(u\) lie in the same layer and have the same number of children in their
respective rooted subtrees, i.e., \(|N_{T^v}(v)| = |N_{T^u}(u)|\). Let the children of \(v\) be 
\(v_1, v_2, \ldots, v_{d-1}\), ordered so that 
\[
T^{v_1} \preceq T^{v_2} \preceq \cdots \preceq T^{v_{d-1}}.
\]
Similarly, list the children of \(u\) as \(u_1, u_2, \ldots, u_{d-1}\), ordered according to
$\preceq$.

Now consider the smallest index \(i\) for which \(T^{v_i}\) and \(T^{u_i}\) are not
isomorphic. If \(T^{v_i} \prec T^{u_i}\), then comparator outputs \(T^v \prec T^u\), and vice versa.
\end{description}

\begin{lemma}
    The subtree comparator is well-defined in the following sense: (1) the comparator does not recurse infinitely and (2) the comparator is transitive.
\end{lemma}

\begin{proof}
 To show that the recursion terminates, observe that recursion is required only when comparing $T^v$ and $T^u$ where $v$ and $u$ lie in the same layer, say layer $i$. By definition of a subtree, every vertex in $T^v \setminus \{v\}$ belongs to a layer lower than that of $v$ (and the same holds for $T^u$ and $u$). Hence, all recursive comparisons involving 
\[
T^{v_1}, T^{v_2}, \ldots, T^{v_{d-1}}, \qquad T^{u_1}, T^{u_2}, \ldots, T^{u_{d-1}}
\]
involve subtrees rooted at vertices in layers at most $i-1$. This ensures that the comparator performs only finitely many recursive calls.

 We now prove transitivity of the comparator. Given distinct subtrees $T^u$, $T^v$, and $T^w$, assume that the comparator outputs $T^u \prec T^v$ and $T^v \prec T^w$. We claim that it must also output $T^u \prec T^w$.

If $u$, $v$, and $w$ do not lie in the same layer, the conclusion is immediate, since $u$ must be in a layer lower than $w$. A similar argument applies if $u$, $v$, and $w$ do not have the same number of children. Thus, the only remaining case is when both comparisons $T^u \prec T^v$ and $T^v \prec T^w$ were decided by the \textbf{Recursion rule}. 

In this case, let $u_1, u_2, \ldots, u_{d-1}$ be the children of $u$ ordered according to~$\preceq$, and similarly for $v$ and $w$. Let $i$ be the smallest index for which $T^{u_i}$ is not isomorphic to $T^{v_i}$, and let $j$ be the smallest such index for the comparison between $T^v$ and $T^w$. Then $k = \min(i, j)$ is the smallest index at which $T^{u_k}$ is not isomorphic to $T^{w_k}$. Moreover, at least one of $T^{u_k} \prec T^{v_k}$ or $T^{v_k} \prec T^{w_k}$ holds, and in either case we obtain $T^{u_k} \prec T^{w_k}$. Applying the \textbf{Recursion rule} once more then yields $T^u \prec T^w$. This establishes transitivity.
\end{proof}

\paragraph{Subtree enumeration} Now all nodes obtain identical enumerations.

\begin{equation*}
    T_1, T_2,\ldots T_k,
\end{equation*}
where $k$ is the number of distinct (non-isomorphic) subtrees appearing in $G$.

For convenience of presentation, we define a helper function that maps each vertex $v$ to the index of the subtree rooted at $v$ in the enumerated list.

\begin{definition}
    For a vertex \(v\), define
    \[
        \tau(v) = i \quad \text{such that } T^v \text{ is isomorphic to } T_i.
    \]
\end{definition}
We also define a helper function that maps the index $i$ to the integer $k-i$.

\begin{definition}
    Define
    \begin{equation*}
        \lambda(i) = k-i.
    \end{equation*}
\end{definition}

Write $\lambda\tau(v)$ as shorthand for $\lambda(\tau(v))$. By definition, we have $\lambda\tau(v) < \lambda\tau(u)$ if and only if $T^v \prec T^u$, and $\lambda\tau(v) = \lambda\tau(u)$ if and only if $T^v$ is isomorphic to $T^u$. 

We are now ready to present the rules for our algorithm. Each node derives three categories of rules from the subtree enumeration: \ruleu, \rulel, and \ruled. Initially, only \ruleu~and \rulel~rules are active. After a node triggers at least one of \ruleu~and \rulel~rules, \ruled~rules become active for that node.

\begin{mdframed}
\begin{itemize}
    \item \ruleu: One rule is constructed for each enumerated subtree, \emph{except for the last one}. For each $T_i$ with $i \le k-1$, let the neighbors of the root be $m_1, m_2, \ldots, m_{d-1}$. Create the rule 
\[
    [\lambda\tau(m_1), \lambda\tau(m_2), \ldots, \lambda\tau(m_{d-1})] \;\rightarrow\; \lambda(i).
\]
This rule is interpreted as follows: whenever a node with $d$ ports observes that $d-1$ of its ports have respectively received 
\[
    \lambda\tau(m_1), \lambda\tau(m_2), \ldots, \lambda\tau(m_{d-1})
\]
pulses, while the remaining port has received \emph{none}, it marks that remaining port as $\mathsf{Port*}$ and sends pulses along $\mathsf{Port*}$ until the total number of pulses sent through it is exactly $\lambda(i)$. The node then adds \ruled\ to its list of active rules and disables the \rulel\ rule.

Note that in the special case of $T_1$, the subtree consists of a single leaf node. The rule generated by $T_1$ is therefore equivalent to the following: if a node has only one port, it sends $\lambda(1) = k-1$ pulses along that port \emph{immediately upon initialization}.
    
    \item \rulel~(Odd diameter): (This rule is added if $G$ has odd diameter.) For $T_k$, consider the neighbors of the root $m_1, m_2, \ldots, m_{d-1}$.   
The rule is as follows: whenever a node with $d$ ports sees that $d-1$ of its ports have respectively received
\[
    \lambda\tau(m_1),\, \lambda\tau(m_2),\, \ldots,\, \lambda\tau(m_{d-1})
\]
pulses, \emph{while the remaining port has received one pulse}, the node sends a pulse along all ports, declares $\leader$, and terminates. 

    \item \rulel~(Even diameter): (This rule is added if $G$ has even diameter.) For $T_k$, consider the neighbors of the root $m_1, m_2, \ldots, m_d$. The rule is as follows: whenever a node with $d$ ports sees \emph{all of its ports} receive, respectively, 
\[
    \lambda\tau(m_1),\ \lambda\tau(m_2),\ \ldots,\ \lambda\tau(m_d)
\]
pulses, the node sends a pulse along all ports, declares $\leader$, and terminates.
    \item \ruled: When a pulse is received from $\mathsf{Port*}$, the node sends a pulse along every port except $\mathsf{Port*}$, declares $\nonleader$, and terminates.
\end{itemize}
\end{mdframed}

\paragraph{Manual parent assignment} Recall that we defined parent and children in \Cref{def:parent and children}. If the advised topology $G$ has even diameter, then $V_r = \{l\}$ is a singleton set; apart from $l$, each vertex has its parent assigned, and we aim to elect the node corresponding to $l$ as the leader. If $G$ has odd diameter, $V_r$ contains two vertices $l_1$ and $l_2$, which both have no parents. The fact that $G$ is not symmetric over any edge enables us to decide a leader between $l_1$ and $l_2$: in fact, $T^{l_1}$ must not be isomorphic to $T^{l_2}$, and they are exactly the last two trees enumerated due to belonging to the highest layer $r$. Without loss of generality, we assume $T^{l_1} = T_{k-1}$, $T^{l_2} = T_k$ in the remainder of this section. To make proofs more concise, we manually assign $l_2$ as the parent of $l_1$, and $l_1$ as a child of $l_2$. At no risk of confusion, if the diameter is odd, define $l = l_2$. In either case, the only vertex \emph{without a parent} is $l$, and we call it the \emph{root} of the tree. Our goal is to elect the node corresponding to the root vertex $l$ as leader.

\subsubsection{Proof of Correctness} \label{sect:correctnessproof}

We now prove the correctness of the algorithm.

\begin{lemma} [\ruleu~rules are well-defined]
\label{lem: partial_transmission}
Consider any two distinct \textbf{Upstream} rules, generated by $T_j$ and $T_{j'}$, whose roots each have the same number of neighbors: 
\begin{equation*}
    [\lambda\tau(m_1), \lambda\tau(m_2), \ldots, \lambda\tau(m_{d-1})] \rightarrow \lambda(j) \quad \text{and} \quad
    [\lambda\tau(m'_1), \lambda\tau(m'_2), \ldots, \lambda\tau(m'_{d-1})] \rightarrow \lambda(j').
\end{equation*}
Without loss of generality, assume each list of $\lambda\tau$-values is sorted in non-increasing order. If 
\[
\lambda\tau(m_1) \ge \lambda\tau(m'_1),\;
\lambda\tau(m_2) \ge \lambda\tau(m'_2),\;
\ldots,\;
\lambda\tau(m_{d-1}) \ge \lambda\tau(m'_{d-1}),
\]
then $\lambda(j) > \lambda(j')$.
\end{lemma}

\begin{proof}
The definition of $\lambda$ implies that
\[
\tau(m_1) \le \tau(m'_1),\ 
\tau(m_2) \le \tau(m'_2),\ \ldots,\ 
\tau(m_{d-1}) \le \tau(m'_{d-1}).
\]
To prove $\lambda(j) > \lambda(j')$, it suffices to show that $j < j'$, i.e., that $T_j$ is enumerated before $T_{j'}$. Since the respective roots of $T_j$ and $T_{j'}$ have the same degree, the \textbf{Fewer children first} rule does not apply. Thus, the comparison must be determined by either the \textbf{Low layer first} rule or the \textbf{Recursion rule}.

If the root of $T_j$ lies in a lower layer than the root of $T_{j'}$, then $T_j \prec T_{j'}$ immediately by the \textbf{Low layer first} rule.
    
Now consider the case where the two roots lie in the same layer. Since the \ruleu\ rules under consideration are distinct, there exists a smallest index $i$ such that 
\[
\tau(m_i) < \tau(m'_i),
\]
meaning that $T^{m_i}$ is enumerated before $T^{m'_i}$, and hence $T^{m_i} \prec T^{m'_i}$. By the \textbf{Recursion rule}, we then obtain $T_j \prec T_{j'}$, so $T_j$ is enumerated first.

It remains to show that it is impossible for the root of $T_j$ to lie in a \emph{higher} layer than the root of $T_{j'}$. Since each list of $\lambda\tau$-values is sorted in non-increasing order, we have
\[
\tau(m_1) \le \tau(m_2) \le \cdots \le \tau(m_{d-1}).
\]
Thus $T^{m_{d-1}}$ is the last among these subtrees to be enumerated, so $m_{d-1}$ is a highest-layer child of the root of $T_j$. By \Cref{obs:one child in layer i-1}, $m_{d-1}$ must lie exactly one layer below the root of $T_j$, and therefore must lie in a higher layer than $m'_{d-1}$. Hence 
\[
\tau(m_{d-1}) > \tau(m'_{d-1}),
\]
contradicting the fact that $\tau(m_{d-1}) \le \tau(m'_{d-1})$.   
\end{proof}

The following observation still applies.

\portconstant*

\Cref{lem: partial_transmission} and \Cref{obs: port* constant} show that our \textbf{Upstream} rules are \emph{well-defined} in the following sense: While a node can trigger multiple \ruleu~rules throughout the election process, an \ruleu~rule triggered earlier does not send more pulses along \ps~than an \ruleu~rule triggered later.

\begin{lemma} \label{lem: at most lambda tau(u) pulses to parent}
    Let $u$ be a node with $u \neq l$. Then $u$ sends at most $\lambda\tau(u)$ pulses to its parent.
\end{lemma}

\begin{proof}
Recall from the definition of \ruleu~that if the rule constructed from the $\tau(u)$th enumerated tree $T_{\tau(u)} = T^u$ is triggered, $\lambda\tau(u)$ pulses are sent to \ps. The only possible way for $u$ to send additional pulses would be to trigger another \ruleu\ rule and send $\lambda\tau(u') > \lambda\tau(u)$ pulses to $v$. We will show that this cannot occur. We proceed by induction over layers, from $V_0$ to $V_r$, i.e., in a bottom-up manner.

\paragraph{Base case}
Let $u$ be a leaf node in $V_0$. Then $T^u$ is a single-node tree and must be the first to be enumerated, i.e., $\tau(u) = 1$. By the \textbf{Upstream} rule constructed from it, $u$ sends $k-1 = \lambda\tau(u)$ pulses to its parent immediately upon initialization and never sends any additional pulses. Hence, the lemma holds for $u$.

\paragraph{Induction step} Assume the lemma holds for all nodes in $V_s$ and lower layers for some $s \in [0, r-1]$. Let $u$ be a node in $V_{s+1}$ with $u \neq l$, and let $v$ be its parent. Note that $\tau(u) \leq k-1$, so $\lambda\tau(u) \geq 1$.
    
    Suppose $u$ sets $\mathsf{Port*}$ to point to $v$ (otherwise the lemma follows immediately for $u$). Let the children of $u$ be $\{m_1, m_2, \ldots, m_{d-1}\}$ with
\[
\lambda\tau(m_1) \ge \lambda\tau(m_2) \ge \cdots \ge \lambda\tau(m_{d-1}).
\]
Assume, for the sake of contradiction, that $u$ eventually triggers some \textbf{Upstream} rule sending $\lambda\tau(u') > \lambda\tau(u)$ pulses to $v$. Let the children of $u'$ be $\{m'_1, m'_2, \ldots, m'_{d-1}\}$, ordered so that
\[
\lambda\tau(m'_1) \ge \lambda\tau(m'_2) \ge \cdots \ge \lambda\tau(m'_{d-1}).
\]
Since each child of $u$ lies in a layer below $V_{s+1}$, the induction hypothesis implies that each $m_i$ sends at most $\lambda\tau(m_i)$ pulses to $u$. Consequently,
\[
\lambda\tau(m_1) \ge \lambda\tau(m'_1),\;
\lambda\tau(m_2) \ge \lambda\tau(m'_2),\;
\ldots,\;
\lambda\tau(m_{d-1}) \ge \lambda\tau(m'_{d-1}).
\]
By \Cref{lem: partial_transmission}, we have $\lambda\tau(u') < \lambda\tau(u)$, contradicting the assumption that $\lambda\tau(u') > \lambda\tau(u)$.
\end{proof}

\begin{lemma} \label{lem: general tree send to parent before receiving}
    Let $u$ be a node with $u \neq l$. Then $u$ does not receive any pulses from its parent $v$ before sending one to $v$. 
\end{lemma}
\begin{proof}
We first show that $l$ cannot trigger any \ruleu\ rule. Recall that $T^l = T_k$ is the last enumerated subtree; let the children of $l$ be $\{m_1, m_2, \ldots, m_d\}$, ordered so that $\lambda\tau(m_1) \ge \lambda\tau(m_2) \ge \cdots \ge \lambda\tau(m_d)$. By \Cref{lem: at most lambda tau(u) pulses to parent}, node $l$ receives at most $\lambda\tau(m_1), \lambda\tau(m_2), \ldots, \lambda\tau(m_d)$ pulses from its children. Consider any \ruleu\ rule associated with some node $u'$ of degree $d$, whose children are $\{m'_1, m'_2, \ldots, m'_{d-1}\}$, with $\lambda\tau(m'_1) \ge \lambda\tau(m'_2) \ge \cdots \ge \lambda\tau(m'_{d-1})$.

Suppose first that $u'$ lies in a layer lower than or equal to $r-1$. Then, by \Cref{obs:one child in layer i-1}, the child $m'_{d-1}$ must lie in a layer at least as high as $r-2$. Meanwhile, regardless of whether the diameter is even or odd, the two children $m_{d-1}$ and $m_d$ of $l$ lie in layers at least as high as $r-1$. Hence,
\[
\lambda\tau(m'_{d-1}) > \lambda\tau(m_{d-1}) \ge \lambda\tau(m_d),
\]
which shows that $l$ cannot trigger the \ruleu\ rule for the subtree rooted at $u'$.

Now suppose instead that $u'$ also lies in layer $r$. Then $T^{u'} = T_{k-1}$, and moreover the last child of $l$, namely $m_d$, is exactly $u'$. By the \textbf{Recursion rule} of the subtree comparator, $T^{u'}$ is ordered before $T^l$, which implies there exists a smallest index $i$ such that $\lambda\tau(m_i) < \lambda\tau(m'_i)$. Thus $l$ again cannot satisfy the \ruleu\ rule for the subtree rooted at $u'$.




    
Next, we proceed by induction over the layers to prove the lemma,  from $V_r$  to $V_0$, i.e., in a top-down manner.

\paragraph{Base case}
In $V_r$, the (possibly) only node we need to consider is $l_1$, whose parent is $l$. Suppose, for the sake of contradiction, that $l_1$ receives a pulse from $l$ before sending one to $l$. Such a pulse cannot result from a triggered \ruleu\ rule, as established above. It also cannot arise from a triggered \rulel\ rule, since that would require $l$ to have received pulses from \emph{all} its ports, yet $l_1$ has not sent any pulses to $l$. This yields a contradiction.

\paragraph{Induction step}
Assume that the lemma holds for all nodes in $V_s$ and higher layers for some $s \in [1, r]$. Let $u$ be a node in $V_{s-1}$ and let $v$ be its parent. Suppose, for the sake of contradiction, that $u$ receives a pulse from $v$ before sending one to $v$. As in the base case, $v$ cannot trigger the \rulel\ rule. Thus, it remains to rule out the possibility that $v$ triggers an \ruleu\ rule and sets the port pointing to $u$ as \ps. 

If $v = l$, then the same argument as in the base case yields a contradiction. If instead $v$ has a parent $w$, then for $v$ to trigger \ruleu, it must have received a non-zero number of pulses from $w$. By the induction hypothesis, $v$ must have sent pulses to $w$ before receiving any, and such pulses can only result from an \ruleu\ rule with the port pointing to $w$ designated as \ps. This contradicts \Cref{obs: port* constant}. Hence, $u$ also satisfies the lemma. 
\end{proof}

Observe that the following four lemmas continue to apply. We omit the proofs, as they follow from arguments nearly identical to those used previously.

\onlylmaybeleader*
\portpointsparent*
\leaderbeforedownstream*
\pulsedirection*




The next lemma follows from the way we constructed \ruleu\ and \rulel\ rules.

\begin{lemma} \label{obs:necessary condition for leader rule}
Node $l$ \ul{may} trigger the \rulel\ rule \ul{only after} every node $v \ne l$ has triggered the rule corresponding to its subtree $T_{\tau(v)} = T^v$ and has set $\mathsf{Port*}$ to point to its parent.
\end{lemma}

\begin{proof}
\Cref{lem: even diameter pulse direction}  establishes that, before $l$ triggers \rulel, pulses travel from children to parents. Therefore, we only need to consider pulses directed as such, as we are reasoning about a \ul{necessary condition} for $l$ to trigger \rulel.

\paragraph{Base case} First, let the children of $l$ be $\{m_1, m_2, \ldots, m_d\}$, ordered so that  
\[
\lambda\tau(m_1) \ge \lambda\tau(m_2) \ge \cdots \ge \lambda\tau(m_d).
\]
By the definition of the \rulel\ rule, $l$ must receive 
\[
\lambda\tau(m_1),\, \lambda\tau(m_2),\, \ldots,\, \lambda\tau(m_d)
\]
pulses from its children, respectively, in order to trigger \rulel. Since \Cref{lem: at most lambda tau(u) pulses to parent} shows that each child $m_i$ can send at most $\lambda\tau(m_i)$ pulses to $l$, it follows (by an induction on $i$ from $d$ down to $1$) that $l$ \emph{may} trigger the \rulel\ rule \emph{only after} each child $m_i$ has triggered the \ruleu\ rule corresponding to its own subtree $T_{\tau(m_i)} = T^{m_i}$. 
    
\paragraph{Induction step}
Let the children of a node $v \neq l$ be $\{m_1, m_2, \ldots, m_{d-1}\}$, ordered so that  
\[
\lambda\tau(m_1) \ge \lambda\tau(m_2) \ge \cdots \ge \lambda\tau(m_{d-1}).
\]
Assume it has been established that $l$ \emph{may} trigger the \rulel\ rule \emph{only after} $v$ triggers the \ruleu\ rule defined by $T_{\tau(v)} = T^v$.

By the definition of the \ruleu\ rule constructed from $T_{\tau(v)} = T^v$, node $v$ must receive 
\[
\lambda\tau(m_1),\; \lambda\tau(m_2),\; \ldots,\; \lambda\tau(m_{d-1})
\]
pulses from its $d-1$ children in order to trigger the \ruleu\ rule of $T_{\tau(v)} = T^{v}$. Meanwhile, by \Cref{lem: at most lambda tau(u) pulses to parent}, each child $m_i$ can send \emph{at most} $\lambda\tau(m_i)$ pulses to $v$. Similarly, it follows that $v$ \emph{may} trigger the \ruleu\ rule corresponding to $T_{\tau(v)} = T^{v}$ \emph{only after} each child $m_i$ has triggered its own \ruleu\ rule associated with $T_{\tau(m_i)} = T^{m_i}$. 

Therefore, we have extended the {necessary condition} for $l$ to trigger \rulel\ from node $v \neq l$ to all children of $v$. 
\end{proof}

The above lemma does not guarantee that $l$ \emph{will} trigger the \rulel\ rule. It only points out a condition for $l$ to trigger the \rulel\ rule. To show that this condition is eventually satisfied and that $l$ indeed triggers the \rulel\ rule, we present the next lemma.

\begin{lemma} \label{lem: general tree success}
Every node $v$ other than $l$ eventually triggers the \ruleu\ rule corresponding to its subtree $T_{\tau(v)} = T^v$ and sets $\mathsf{Port*}$ to point to its parent. After this has occurred for all nodes $v \ne l$, node $l$ will trigger the \rulel\ rule. At the moment \rulel\ is triggered, the tree is in quiescence, i.e., no pulses are in transit.
\end{lemma}

\begin{proof}
    We again apply induction over layers, from $V_0$ to $V_r$, to show that every node $u$ other than $l$ eventually triggers the \ruleu\ rule corresponding to its subtree $T_{\tau(u)} = T^u$ and sets $\mathsf{Port*}$ to point to its parent.
    To do so, it suffices to show that $u$ eventually sends $\lambda\tau(u)$ pulses to its parent, as  the \ruleu\ rules are well-defined (\Cref{lem: partial_transmission} and \Cref{obs: port* constant}). 

Indeed, once every node $v \ne l$ has sent $\lambda\tau(v)$ pulses to its parent, node $l$ receives enough pulses from its neighbors to trigger the \rulel\ rule, as required.

\paragraph{Base case}
It is immediate that any node $u \in V_0$ eventually sends $\lambda\tau(u)=\lambda(1)=k-1$ pulses to its parent, as this follows directly from the definition of the \ruleu\ rule applied to degree-1 nodes.

\paragraph{Induction step}
Assume the lemma holds for all nodes up to $V_s$ for some $s \in [0, r-1]$. Let $u$ be a node in $V_{s+1}$. By the induction hypothesis, the children of $u$, denoted
$\{m_1, m_2, \dots, m_{d-1}\}$, eventually send
$\lambda\tau(m_1), \lambda\tau(m_2), \ldots, \lambda\tau(m_{d-1})$ pulses to $u$, respectively. This causes the \ruleu\ rule corresponding to the subtree
$T_{\tau(u)} = T^u$ to trigger, so $u$ sends $\lambda\tau(u)$ pulses through \ps\ to its parent.

During this process, $u$ does not receive any unexpected pulses from its parent. The possibilities of $u$ getting such a pulse from its parent due to \ruleu, \rulel, or \ruled\ are ruled out by combining 
\Cref{lem: even diameter only l may be leader}, 
\Cref{lem: even diameter port* points to parent},
\Cref{lem: even diameter leader before downstream} (or equivalently, just \Cref{lem: even diameter pulse direction}), and \Cref{obs:necessary condition for leader rule}. 

Finally, once the \ruleu\ rule for $T_{\tau(u)} = T^u$ has been triggered, the entire subtree rooted at $u$ reaches quiescence: each child of $u$ has already sent the maximum number of pulses permitted by \Cref{lem: at most lambda tau(u) pulses to parent}, and all such pulses have been received.
\end{proof}

    We remark that the above proof, combined with \Cref{lem: at most lambda tau(u) pulses to parent}, shows that the number of pulses every node $v \ne l$ sends to its parent due to the \ruleu\ rules is \emph{exactly} $\lambda\tau(v)$.  

\begin{lemma} \label{lem: all besides l nonleader}
    All nodes other than $l$ eventually terminate quiescently as $\nonleader$ after $l$ terminates quiescently as $\leader$.
\end{lemma}

\begin{proof}
We first observe that by \Cref{lem: general tree success}, the node $l$ eventually terminates as $\leader$, and at that moment the tree is in quiescence. For all remaining nodes, we assume that $l$ has already terminated as $\leader$ and proceed by induction from $V_r$ down to $V_0$.

\paragraph{Base case} 
For the layer $V_r$, the only node we may need to consider is $l_1$. By \Cref{lem: general tree success}, the node $l_1$ has already triggered the \ruleu\ rule corresponding to $T_{\tau(l_1)}$ and set $\mathsf{Port*}$ to point to $l$ before $l$ terminates as $\leader$ and sends a pulse to $l_1$. Consequently, $l_1$ will eventually receive a pulse from $\mathsf{Port*}$, trigger the \ruled\ rule, and terminate as $\nonleader$. At this point, $l$ has already terminated and will not send further pulses to $l_1$, and by \Cref{lem: at most lambda tau(u) pulses to parent}, $l_1$ will not receive additional pulses from its children. Thus, $l_1$ reaches quiescent termination.

\paragraph{Induction step}
Assume the lemma holds for all nodes in $V_s$ or higher layers, for some $s \in [1, r]$. Let $u$ be a node in $V_{s-1}$. By \Cref{lem: general tree success}, $u$ must have triggered the \ruleu\ rule corresponding to $T_{\tau(u)}$ and set $\mathsf{Port*}$ to its parent $v$. Since $v$ lies in a higher layer, the induction hypothesis ensures that $v$ eventually triggers the \ruled\ rule and terminates as $\nonleader$. Therefore, $u$ will eventually receive a pulse from $\mathsf{Port*}$, trigger the \ruled\ rule, and terminate as $\nonleader$. The argument for quiescent termination follows exactly as in the base case.
\end{proof}

\begin{lemma} \label{lem: general tree complexity}
    The total number of pulses sent by all nodes in this algorithm is $O(n^2)$.
\end{lemma}

\begin{proof}
For each node other than $l$, exactly $\lambda\tau(u) \le k-1 \le n-1$ pulses are sent to its parent due to the \ruleu\ rules (see the proofs of \Cref{lem: at most lambda tau(u) pulses to parent,lem: general tree success}), contributing at most $(n-1)^2$ pulses in total. The \rulel\ and \ruled\ rules cause exactly one pulse to be sent across each edge (see the proof of \Cref{lem: all besides l nonleader}), contributing an additional $n-1$ pulses. Therefore, the total number of pulses sent during the execution of the algorithm is at most $(n-1)^2 + (n-1) = O(n^2)$.
\end{proof}

\mainthm*
\begin{proof}[Proof of \Cref{thm:general_tree_leader_election}]
The correctness and quiescent termination of the leader election algorithm follow from \Cref{lem: general tree success,lem: all besides l nonleader}, and the message complexity bound follows from \Cref{lem: general tree complexity}.
\end{proof}

\section{Randomized Impossibilities}\label{sect:impossibility}
In this section, we show the following impossibility result: 

\impossibility*

    Let us first consider a 2-node path $P_2$, where the two nodes are $\{u,v\}$. Consider the view of $u$ (resp.,~$v$): since the node has only one communication channel, and the communication channel transmits only indistinguishable pulses, the algorithm restricted to $u$ can be captured by an automaton with input alphabet size 1, i.e.,~a probabilistic unary automaton $\{Q,q_s,F_{\leader},F_{\nonleader}, \delta,\sigma\}$ as follows. Note that $u$ must reach some accepting state to terminate, and the output of $u$, $\leader$ or $\nonleader$ depends on the exact accepting state it reaches.
    

    \begin{mdframed}
    \begin{itemize}
        \item $Q$ is the state space.
        \item $q_s$ is the initial state.
        \item $F_{\leader} \subseteq Q$ is the set of accepting states labeled $\leader$.
        \item $F_{\nonleader} \subseteq Q$ is the set of accepting states labeled $\nonleader$.
        \item $\delta: Q\times Q \rightarrow[0,1]$ is the transition function. $\delta(q_1,q_2)$ denotes the probability that the automaton transitions into state $q_2$ after receiving a symbol when in state $q_1$. It holds that $\sum_{q_2\in Q}\delta(q_1,q_2) = 1$. If $q_1 \in F_{\leader}$, $\delta(q_1,q_2) > 0$ implies $q_2 \in F_{\leader}$. If $q_1 \in F_{\nonleader}$, $\delta(q_1,q_2) > 0$ implies $q_2 \in F_{\nonleader}$. This requirement stems from termination of the algorithm: after outputting $\leader$ or $\nonleader$ (denoted by transitioning into a state in $F_{\leader}$ and $F_{\nonleader}$), a node cannot change its output.
        \item $\sigma: Q \rightarrow \mathbb{Z}_{\geq 0}$. When the automaton transitions into state $q$, the node sends $\sigma(q)$ pulses.
    \end{itemize}
    \end{mdframed}

    Let $P_{\leader}(c)$ denote the probability that the automaton's state is in $F_{\leader}$ after receiving $c$ symbols, and define $P_\nonleader(c)$ in a similar fashion.

\begin{lemma} \label{lem: P monotonous}
    For $c_1 \leq c_2$, the following hold.
    \begin{itemize}
        \item $P_{\leader}(c_1) \leq P_{\leader}(c_2)$.
        \item $P_{\nonleader}(c_1) \leq P_{\nonleader}(c_2)$.
    \end{itemize}
\end{lemma}

\begin{proof}
    We only need to show that for $c_2 = c_1+1$, the statements are true. Since $F_{\leader} \subseteq Q$ are the \emph{accepting} states, upon receiving $c_1$ symbols, there are two cases.
    \begin{itemize}
        \item If the automaton is in a state in $F_{\leader}$, then upon receiving the next letter, the automaton is still in a state in $F_{\leader}$.
        \item If the automaton is not in a state in $F_{\leader}$, then upon receiving the next letter, the automaton may or may not make a transition to a state in $F_{\leader}$.
    \end{itemize}
    Therefore, $P_{\leader}(c_1) \leq P_{\leader}(c_2)$. The same argument works for the $\nonleader$ side.
\end{proof}

Now we define $P_{\leader} = \sup_{c\in \mathbb{Z}_{\geq 0}}P_{\leader}(c)$, and $P_{\nonleader}$ is defined similarly.

\begin{lemma} \label{lem: sup bound}
    $P_{\leader} + P_{\nonleader} \leq 1$.
\end{lemma}

\begin{proof}
    Assume otherwise $P_{\leader}+ P_{\nonleader}> 1$. Then there exist $c_{\leader}, c_{\nonleader}$ such that $P_{\leader}(c_{\leader}) + P_{\nonleader}(c_{\nonleader}) > 1$. Without loss of generality, let $c_{\nonleader} \geq c_{\leader}$. By \Cref{lem: P monotonous}, $P_{\leader}(c_{\nonleader}) + P_{\nonleader}(c_{\nonleader}) > 1$, which is impossible.
\end{proof}

We are now ready to prove \Cref{thm: symmetric_about_edge}.

\begin{proof}[Proof for \Cref{thm: symmetric_about_edge}]
We first consider the two-party case, where two nodes $\{u,v\}$ performing some leader election algorithm via content-oblivious communication along an edge. To upper bound the success probability, it suffices to assume that the communication between  $\{u,v\}$ stabilize at some moment (i.e.,~there are no more pulses traveling between $u$ and $v$), for otherwise $u$ and $v$ will never terminate, failing the leader election task. Let $Q(c_u, c_v)$ denote the probability that at stabilization, $u$ received $c_u$ pulses, and $v$ received $c_v$ pulses. The probability that the leader election is successful can be obtained as follows: For all possibilities of $(c_u,c_v)$ pair, sum up the probability that the outputs of $u$ and $v$ constitute a valid leader election outcome, weighted by $Q(c_u, c_v)$.
\begin{align*}
    \Pr[\text{Correct}]
    &= \sum_{c_v, c_u} Q(c_u, c_v)(P_{\leader}(c_v)P_{\nonleader}(c_u)+P_{\nonleader}(c_v)P_{\leader}(c_u)) \\
    &\leq P_{\leader}P_{\nonleader} + P_{\nonleader}P_{\leader} \\
    &\leq \frac{1}{2}
\end{align*}
We have thus proven that for an anonymous edge, content-oblivious leader election terminates correctly with probability at most $\frac{1}{2}$. We are now ready to extend the argument to all graphs $G$ symmetric about an edge $\{u,v\}$ by a simple simulation argument.

Let graph $G$ be symmetric about edge $e$, so $G-e$ has two isomorphic connected components, $H$ and $K$.  We claim $G$ does not admit a leader election algorithm that succeeds with probability greater than $\frac{1}{2}$, even if all nodes know the topology $G$. Otherwise, on a 2-node path $P_2$, let the two nodes $u$ and $v$ respectively simulate $H$ and $K$, and perform the algorithm on $G$, which leads to contradiction. Since $H$ and $K$ are isomorphic, the simulation does not require first breaking the symmetry between the two nodes $u$ and $v$.
\end{proof}

The following corollary, which is restricted to the deterministic case, directly follows.

\impossibilitydet*

\begin{proof}
    We want to show for any edge-symmetric network $G$ where all nodes know $G$ a priori, for any leader election algorithm, there exists some $\id$ assignment on which the leader election fails. Assume otherwise: fix any $G$, there is a leader election algorithm that succeeds with all $\id$ assignments. However, this contradicts with \Cref{thm: symmetric_about_edge} by the following reduction. With the knowledge of the size of $G$, the anonymous but randomized nodes, as in the setting of \Cref{thm: symmetric_about_edge}, each samples uniformly from an $\id$ space $\{1,2,\ldots,3|V(G)|^2\}$. There are $\binom{|V(G)|}{2} < |V(G)|^2$ node pairs, each bearing a $\frac{1}{3|V(G)|^2}$ probability of $\id$ clashing. By a union bound, with probability at least $\frac{2}{3}$, all nodes choose different $\id$s. In such a case, the deterministic leader election algorithm can be invoked. This reduction gives a leader election algorithm with success probability at least $\frac{2}{3}$ for anonymous and randomized nodes, which contradicts the $\frac{1}{2}$ upper bound of \Cref{thm: symmetric_about_edge}. Therefore, the assumption is erroneous, so for any edge-symmetric network $G$, no leader election algorithm succeeds with all $\id$ assignments.
\end{proof}

\paragraph{Discussion}
The proof of our impossibility result is similar to the two-party impossibility in \cite[Theorem 20]{content-oblivious-leader-election-24}. In both proofs, the argument uses the contentless nature of pulses to model the behavior of degree-1 nodes as a unary automaton. In our case, we allow this automaton to be probabilistic, which yields a slightly stronger impossibility. The key difference between the two results lies in the reduction that extends the impossibility from a single edge to graphs containing bridges.

In our setting, where nodes are given topology knowledge, the assumption that the graph is \emph{non-edge-symmetric} is essential. Suppose there exists a leader election algorithm that succeeds on all (or some) non-edge-symmetric graphs $G$. The existence of such an algorithm does \emph{not} appear to contradict either the deterministic two-party impossibility of \cite[Theorem 20]{content-oblivious-leader-election-24} or our randomized two-party impossibility. For two parties $u$ and $v$ to simulate a leader election algorithm on $G$, they would need to simulate two \emph{non-isomorphic} connected components $H$ and $K$ of $G - e$ for some edge $e$. However, deciding which party simulates $H$ and which simulates $K$ already requires a leader election between $u$ and $v$. If $u$ and $v$ choose $H$ and $K$ independently, the probability that their choices differ is at most $1/2$.



\section{Necessity of Knowledge on Underlying Topology}\label{sect:necessity}

In this section, we show the necessity of exact topology knowledge for leader election on trees. We show a family $\mathcal{G}$ of two graphs where each $G \in \mathcal{G}$ by itself grants a quiescently terminating algorithm, but no algorithm achieves termination if the graph topology is drawn from $\mathcal{G}$.


We write $P_n$ to denote a path of $n$ nodes. We prove the following: 

\necessity*

Since $P_3$ and $P_5$ are not edge-symmetric, if the graph topology is known, then we do have a quiescently terminating algorithm for them, as shown in \Cref{thm:general_tree_leader_election}.

\begin{proof}
    For the sake of contradiction, we assume such a deterministic, terminating algorithm $\mathcal{A}$ exists. We first define a labeling over the nodes in $P_n$.
    
    \begin{itemize}
        \item Assign label $\lnode$ to the nodes with degree 1 (leaf nodes).
        \item Assign label $\mnode$ to the nodes with degree 2 (middle nodes).
    \end{itemize}
    Since a node can count the number of its ports, it can label itself at the initiation of any algorithm, so we may assume that the labels are provided as input to the nodes. Therefore, $\mathcal{A}$ can be viewed as a deterministic, terminating leader election algorithm $\mathcal{A'}$ for the family of labeled graphs:
    \begin{equation*}
        \{\lnode-\mnode-\lnode \hspace{2mm},\hspace{2mm}\lnode-\mnode-\mnode-\mnode-\lnode\}
    \end{equation*}

    Now we define a new type of labeled node $\xnode$. Observe that the algorithm $\mathcal{A'}$ can be turned into an algorithm $\mathcal{A''}$ that elects a leader for the following family of labeled graphs:
    \begin{equation*}
        \{\lnode-\mnode-\lnode \hspace{2mm},\hspace{2mm}
        \lnode-\mnode-\mnode-\mnode-\lnode \hspace{2mm},\hspace{2mm}
        \xnode-\lnode \hspace{2mm},\hspace{2mm}
        \xnode-\mnode-\xnode\}
    \end{equation*}

    To see this, we simply let $\mathcal{A''}$ on all nodes labeled $\xnode$ simulate $\lnode-\mnode-$, and then perform $\mathcal{A'}$. A subtle issue is that since one $\xnode$ node needs to simulate two nodes, $\xnode$ nodes must ensure the simulated nodes have $\id$s that do not conflict with the rest of the graph. To achieve this, when performing $\mathcal{A''}$, $\lnode$ and $\mnode$ nodes with assigned $\id$ $k$ runs $A'$ with $\id$ $2k$, while a $\xnode$ simulates an $\lnode$ node of $\id$ $2k$ and an $\mnode$ node of $\id$ $2k+1$.

    Similar to our approach in \Cref{thm: symmetric_about_edge} and \cite[Theorem 20]{content-oblivious-leader-election-24}, the behavior of an $\lnode$ node with assigned $\id$ $i$ performing algorithm $\mathcal{A''}$ can be described by an unary automaton $\mathcal{A''}_{\lnode}(i)$. The state of the unary automaton $\mathcal{A''}_{\lnode}(i)$ is entirely a deterministic function of the number of symbols received, while the terminating requirement of $\mathcal{A''}$ forbids it from changing output (from $\leader$ to $\nonleader$ or the opposite direction). Therefore, the output of $\mathcal{A''}_{\lnode}(i)$ is entirely a function of $i$. Now consider all such unary automata:
    \begin{equation*}
        \mathcal{A''}_{\lnode}(1), \mathcal{A''}_{\lnode}(2), \mathcal{A''}_{\lnode}(3), \ldots  
    \end{equation*}
    for all $\id$s. There must \emph{not} be 2 $\id$s $i$ and $j$ where $\mathcal{A''}_{\lnode}(i)$ and $\mathcal{A''}_{\lnode}(j)$ both output $\leader$, since otherwise we have a contradiction by assigning $\id$s $i$ and $j$ to the two $\lnode$ nodes in the graph $\lnode-\mnode-\lnode$. Hence, there exists an $\id$ $k_{\lnode}$ such that all $\lnode$ nodes with $\id \geq k_{\lnode}$ output $\nonleader$. Similarly, we can define $k_{\xnode}$ due to the presence of the labeled graph $\xnode-\mnode-\xnode$.

    Now for the graph $\xnode-\lnode$ which $\mathcal{A''}$ claims to solve with termination, assign $\id$s $\max(k_{\lnode}, k_{\xnode})$ and $\max(k_{\lnode}, k_{\xnode})+1$ to the two nodes respectively. Both nodes must output $\nonleader$ by our choice of $k_{\lnode}$ and $k_{\xnode}$, but they are the only nodes in the graph $\xnode-\lnode$. Hence, we arrive at a contradiction.
\end{proof}

\paragraph{Generalization} Similar to the proof of \Cref{thm: symmetric_about_edge}, we believe that the argument above can be extended to the randomized setting; however, we chose not to pursue this extension for the sake of simplicity.

There exist many graph families beyond $\mathcal{G} = \{P_3, P_5\}$ for which the same impossibility result can be established. Our argument does not depend on the internal structure of $\lnode$; in fact, replacing $\lnode$ with any graph, \emph{not limited to rooted trees}, does not affect the impossibility proof, thereby yielding infinitely many counterexamples. Consequently, our proof of the necessity of topology knowledge applies far more broadly than just to trees.

\section{Stabilizing Leader Election on Trees}\label{sec:Stabilizing}

In this section, we show a simple stabilizing leader election algorithm for trees. By trimming leaves iteratively, this algorithm reduces the problem of stabilizing leader election on a tree to stabilizing leader election on a single edge, and the latter admits a simple algorithm by comparing $\id$s.

\stabilizing*

\begin{algorithm}
\DontPrintSemicolon
\caption{Stabilizing algorithm for node $v$}
\label{alg: tree_stabilizing}
$\mathsf{isLeaf} \leftarrow \mathtt{False}$\;
\OnEvent{$v$ is a leaf and $\mathsf{isLeaf} = \mathtt{False}$}{
    $\mathsf{sendPulse()}$\; \label{send-pulse=1}
    $\mathsf{isLeaf} \leftarrow \mathtt{True}$\;
}
\OnEvent{$\mathsf{receivePulse(\mathit{u})}$ and $\mathsf{isLeaf} = \mathtt{False}$}{
    $\mathsf{removeNeighbor(\mathit{u})}$\;
}
\OnEvent{$\mathsf{receivePulse(\mathit{u})}$ and $\mathsf{isLeaf} = \mathtt{True}$ \label{event}}{
    \textcolor{blue}{\tcp{Edge Leader Election}}
    \For{$i\in[1,\id(v)]$}{
        $\mathsf{sendPulse()}$\; \label{send-pulse=2}
    }
    \For{$i\in[1,\id(v)]$ \label{for-loop}}{
        $\mathsf{receivePulse(\mathit{u})}$\;
    }
    Set itself as the leader and terminate the algorithm\;
}
\end{algorithm}
We first explain routines and variables appearing in \Cref{alg: tree_stabilizing}. 
\begin{itemize}
    \item $\mathsf{isLeaf}$ is a Boolean indicator that is set to $\mathtt{True}$ whenever a node becomes a leaf after removing neighbors using $\mathsf{removeNeighbor}$, or if it is initially created as a leaf.
    \item $\mathsf{sendPulse()}$ sends a single pulse to the node's unique remaining neighbor. Note that $\mathsf{sendPulse()}$ is invoked only by leaf nodes.
    \item $\mathsf{receivePulse(\mathit{u})}$ is triggered whenever a pulse is received from a neighbor $u$.
\end{itemize}
    
Now we prove the correctness and the claimed message complexity of the algorithm.

\begin{proof}[Proof of \Cref{thm:stabilizing}]
Suppose at some point a node $v$ calls $\mathsf{removeNeighbor}(u)$, then at this moment $v$ is not a leaf.
Moreover, $\mathsf{removeNeighbor}(u)$ is triggered only if $u$ has sent a pulse to $v$, meaning that $u$ must have already become a leaf. 

In view of the above, consider the graph $H$ initialized as the original network topology $T$, and whenever a node $v$ calls $\mathsf{removeNeighbor}(u)$, we remove the leaf $u$ from $H$.

 \paragraph{Phase 1: Leaf trimming}  We argue that at some point, $H$ is left with two nodes, $s$ and $t$, connected by an edge.  Assume at the moment $H$ has more than two nodes, among which $k$ are leaf nodes. Let $l_1, l_2, \ldots, l_k$ denote the leaf nodes and $p_1, p_2, \ldots, p_k$ their unique neighbors, respectively. At this moment, there must be one pulse due to Line~\ref{send-pulse=1} traveling along the direction $l_i \rightarrow p_i$ for each $i \in \{1,2,\ldots, k\}$. Consider the moment the first pulse among these pulses arrives at the destination. Without loss of generality, let   $p_1$ be the receiver of this pulse. Since $H$ has more than two nodes, $p_1$ must not be a leaf node. Hence upon receiving the pulse, $p_1$ calls $\mathsf{removeNeighbor}(l_1)$, decreasing the number of nodes of $H$ by one.

On the final edge $\{s,t\}$, there are also two pulses due to Line~\ref{send-pulse=1} traveling in the opposite directions.
For both endpoints $s$ and $t$, the recipt of such a pulse triggers the event of Line~\ref{event}, so they ultimately enter the edge leader election phase.

\paragraph{Phase 2: Edge leader election} We show that the algorithm stablizes with exactly one of $s$ and $t$ as the leader.
Without loss of generality, assume  that $\id(s) >\id(t)$. Then node $t$ will eventually receive exactly $\id(s) >\id(t)$ pulses from $s$ during the for-loop starting at Line~\ref{for-loop}, which is sufficient for $t$ to exit the loop and become the leader. On the other hand, $s$ receives exactly $\id(t) < \id(s)$ pulses from $t$ during the for-loop starting at Line~\ref{for-loop}, meaning that $s$ can never leave the for-loop, so it can never become the leader.

\paragraph{Message complexity} We now analyze the message complexity of the algorithm. Each node sends at most one pulse due to Line~\ref{send-pulse=1}, as there is no way to reset $\mathsf{isLeaf}$ from $\mathtt{True}$ to $\mathtt{false}$. Only the two endpoints of the final edge $\{s,t\}$ can trigger the event of Line~\ref{event} to enter the edge leader election phase, so the total number of pulses sent due to  Line~\ref{send-pulse=2} is $\id(s)+\id(t) \leq 2\mathsf{ID_{max}}-1$. Therefore, the total number of pulses sent is  at most $n+2\cdot\mathsf{ID_{max}}-1$.
\end{proof}

\section{Conclusion and Open Questions}

In this work, we fully characterize the solvability of content-oblivious leader election on trees with topology knowledge: Terminating leader election is possible \emph{if and only if} the tree is not symmetric about an edge.

For \emph{graphs} that are symmetric about an edge, even with topology knowledge and unique identifiers, no randomized leader election can succeed with probability greater than $0.5$. For \emph{trees} that are not symmetric about an edge, knowledge of the network topology enables a quiescently terminating leader election algorithm. Moreover, for even-diameter trees, the requirement of knowing the full topology can be replaced by knowing only the diameter, at the mere cost of losing quiescence. Our algorithms, while consistent with the impossibility result of \citet{fully-defective-22}, demonstrate that non-trivial problems can indeed be solved with a little aid.

The knowledge of the network topology is \emph{necessary} in the following sense:  
Both $P_3$ and $P_5$ are not symmetric about an edge and therefore admit quiescently terminating leader election algorithms; however, if nodes only know that the underlying topology is drawn from $\mathcal{G}=\{P_3, P_5\}$, then we show that terminating leader election is impossible.  


Being the first work to explore content-oblivious leader election beyond 2-edge-connected networks, our focus is primarily on \emph{qualitative} results. It remains widely open whether the efficiency of our leader election algorithms can be improved, both in terms of message complexity and advice complexity. For odd-diameter trees, the advice complexity of our algorithm is $O(n)$, as the topology of an $n$-vertex tree can be encoded using $O(n)$ bits. For even-diameter trees, the advice complexity of our algorithm is $O(\log n)$, as the diameter of an $n$-vertex tree can be written using $O(\log n)$ bits.




Given that content-oblivious leader election in 2-edge-connected networks~\cite{content-oblivious-leader-election-24,2-connected-leader-election-25} and in trees has been studied, a natural next step is to characterize the solvability of content-oblivious leader election in \emph{general graphs}, since any connected graph decomposes into a tree of 2-vertex-connected components, and every 2-vertex-connected graph is also 2-edge-connected. On the impossibility side, we have shown in this work that if $G$ is symmetric about an edge, then terminating leader election is not possible on $G$. We conjecture that, apart from the symmetric case, knowledge of the network topology $G$ suffices to yield a terminating leader election algorithm.


Many graph problems remain unexplored in the content-oblivious model. For 2-edge-connected networks, a universal compiler transforms any distributed algorithm into one in the content-oblivious model, as shown by \citet{fully-defective-22}. Thus, the feasibility of leader election in 2-edge-connected networks implies that every other graph problem is also solvable. In contrast, due to the impossibility result of \citet{fully-defective-22}, for trees and other non-2-edge-connected networks, such a universal simulation cannot exist, even with full topological knowledge. We leave open the question of whether other fundamental graph problems, such as 2-coloring, 3-coloring, and MIS, are solvable on trees in the content-oblivious model, provided that the network topology is known.

\printbibliography
\end{document}